\documentclass{article}

\usepackage{times}
\usepackage{graphicx}
\usepackage{subfig} 
\usepackage{amsfonts}
\usepackage{amsmath} 
\usepackage{amssymb} 
\usepackage{amsthm}
\usepackage{color}

\newcounter{counter} \newtheorem{example}[counter]{Example}
\newcounter{counter2} \newtheorem{lemma}[counter2]{Lemma}
\newcounter{counter3} \newtheorem{theorem}[counter3]{Theorem}
\newcounter{counter4} \newtheorem{corollary}[counter4]{Corollary}

\newcommand{\highlight}[1]{#1}

\begin{document}

\bibliographystyle{plain}

\title{On the Construction of High Dimensional Simple Games}

\author{Martin Olsen\\
BTECH, Aarhus University\\
Denmark\\
martino@btech.au.dk
\and
Sascha Kurz\\
University of Bayreuth\\
Germany\\
sascha.kurz@uni-bayreuth.de
\and
Xavier Molinero\thanks{Research partially
funded by Grant MTM 2012-34426 from the Spanish Economy and Competitiveness Ministry.}\\
Universitat Polit\`ecnica de Catalunya\\
Spain\\
xavier.molinero@upc.edu
}

\date{}

\maketitle

\begin{abstract}
\noindent
Voting is a commonly applied method for the aggregation of the preferences
of multiple agents into a joint decision. If preferences are binary, i.e.,
``yes'' and ``no'', every voting system can be described by a (monotone) Boolean
function $\chi\colon\{0,1\}^n\rightarrow \{0,1\}$. However, its naive
encoding needs $2^n$ bits. The subclass of threshold functions, which is sufficient
for homogeneous agents, allows a more succinct representation using $n$ weights and
one threshold. For heterogeneous agents, one can represent $\chi$ as an
intersection of $k$ threshold functions. Taylor and Zwicker have constructed a sequence
of examples requiring $k\ge 2^{\frac{n}{2}-1}$ and provided a construction guaranteeing
$k\le {n\choose {\lfloor n/2\rfloor}}\in 2^{n-o(n)}$. The magnitude of the worst-case situation was
thought to be determined by Elkind et al.~in 2008, but the analysis unfortunately turned
out to be wrong. Here we uncover a relation to coding theory that allows the determination
of the mi\-nimum number $k$ for a subclass of voting systems. As an application, we give
a construction for $k\ge 2^{n-o(n)}$, i.e., there is no gain from a representation complexity
point of view.  

\medskip

%%\textbf{JEL:} C71, D72
\end{abstract}

\section{Introduction}

Consider a set $N=\{1,\dots,n\}$ of agents whose binary preferences should be aggregated
to a group decision. We assume that voting is used as aggregation method, i.e., each agent
can say either ``yes'' or ``no'', which we encode by $1$ and $0$, respectively, to a given proposal.
The group decision is then an ``accept'' (1) or ``reject'' (0). Formally, the used voting system
can be modeled as a Boolean function $\chi\colon\{0,1\}^n\rightarrow \{0,1\}$. By imposing some,
quite natural, additional constraints, we obtain the class of so-called \emph{simple games}, see
Subsection~\ref{subsec_simple_games}. They are widely applied and very useful tools for understanding
decision making in political and other contexts. One major drawback is that they do not admit an obvious
succinct representation. The naive approach, listing the function values of $\chi$, needs $2^n$ bits.
Listing so-called \emph{minimal winning coalitions}, see Subsection~\ref{subsec_simple_games}, also
needs $2^{n-o(n)}$ items in the worst case.

However, the subclass of threshold functions of monotone Boolean functions or \emph{weighted games} of
simple games, see Subsection~\ref{subsec_simple_games}, can be represented by just $n$ integer weights
$w_i$, for $i\in\{1,\ldots,n\}$, and an integer  threshold/\emph{quota} $q$. If a representation as a threshold function/weighted
games is possible, at most $O(n\log n)$ bits are needed for each integer~\cite{0243.94014}. In the case
of homogeneous agents or players, May's theorem~\cite{may1952set} states that we can choose $w_i=1$ and
$1\le q\le n$, i.e., a very succinct representation is possible. In the case of heterogeneous agents or
players there are unfortunately simple games which are not weighted games if $n\ge 4$. However, Taylor
and Zwicker have constructively shown that each simple game can be represented as the intersection of at
most ${n\choose {\lfloor n/2\rfloor}}\in 2^{n-o(n)}$ weighted games, where the weights are either $0$ or
$1$ and the quotas are $1$, see Subsection~\ref{subsec_simple_games}. The smallest number $k$ such that
a given simple game $\Gamma$ can be represented as the intersection of $k$ weighted games is called the
\emph{dimension} of $\Gamma$. From a representation complexity point of view, we have the following important
questions:
\begin{center}
  How large can the dimension of a simple game be?

  And how can the corresponding weighted games be constructed?
\end{center}

\subsection{Related Work}
With respect to the first question, Taylor and Zwicker provided a sequence of examples requiring at
least $k\ge 2^{\frac{n}{2}-1}$ weighted games \cite{TaZw99}. So, there is a large gap in the knowledge of the
magnitude of the worst-case situation, which was thought to be closed by Elkind et al.~in 2008, see \cite{ElkindGGW08}.
Unfortunately, their analysis is flawed, which we will demonstrate in Section~\ref{TZEx}.

Taylor and Zwicker made the observation that although there are simple games with arbitrarily large dimension,
they do not seem to be used in real-world voting systems. At the time of writing \cite{TaZw99}, the authors were
only aware of practical voting systems with a dimension of at least $2$. Classical examples of dimension $2$ are given
by the Amendment of the Canadian constitution \cite{Kilgour1983} and the US federal legislative system \cite{Taylor1993170}.
The voting systems of the Legislative Council of Hong Kong and the Council of the European Union under its
Treaty of Nice rules have a dimension of exactly three, which was proven in \cite{Cheung2014,Freixas2004415},
respectively. Quite recently, it has been shown that the voting system of the Council of the European Union under its
Treaty of Lisbon rules has a dimension between $7$ and $13\,368$ \cite{Kurz2015}. There, the authors also address
the second question by providing heuristic algorithms based on integer linear programming. Besides that, the probably
first published general approach for the determination of lower bounds for the dimension of a simple game is presented.

Instead of the intersection, each simple game can also be represented as a finite union of weighted games,
which leads to the notion of \emph{co--dimension}, see e.g.~\cite{NotionsOfDimension}. Allowing arbitrary
combinations of unions and intersections results in the concept of a \emph{Boolean dimension}, which is
introduced and studied in \cite{BooleanCombinations}. We remark that the voting system of the Council of the European Union
under Lisbon rules has a Boolean dimension of exactly three \cite{Kurz2015}. For the subclass of so-called
complete simple or linear games, the dimension was studied in \cite{Freixas2001265}.

\subsection{Our Contribution}

We show up a link between the dimension of simple games and co\-ding theory.
More precisely, we give a construction of a simple game from an error-correcting code,
including the determination of the corresponding exact dimension. Using results
on error-correcting codes, we can conclude the existence of simple games whose dimension
asymptotically matches the worst-case upper bound $2^{n-o(n)}$ of Taylor and Zwicker~\cite{TaZw99},
i.e., we close the gap in the literature that was previously filled by the flawed
result of Elkind et al. \cite{ElkindGGW08}.

We thoroughly discuss the lower bound construction of Taylor and Zwicker, i.e, we determine
the corresponding exact dimension. Curiously enough, \highlight{just the integer weights $0$, $1$, and $2$}
are needed for the used weighted games. It turns out that Elkind et al.
considered an isomorphic variant of the example of Taylor and Zwicker.

\subsection{Outline}

The remaining part of this paper is structured as follows:
%
%\bigskip
%\bigskip

In Section~\ref{sec:preliminaries}, we introduce some notation and
formally define the considered concepts in the paper. We also state
a well-known theoretical upper bound for the dimension. % for later use in this section.
Section~\ref{TZEx} shows that the example given by Elkind et al. \cite[Theorem 8]{ElkindGGW08}
is an isomorphic variant of the example given by Taylor and Zwicker \cite[Theorem 1.7.5]{TaZw99}.
The games that form the basis for our results are
introduced in Section~\ref{sec:construction}. Finally,
Section~\ref{sec:hdproof} contains the proofs of high dimension and
a theorem that forms the main contribution of the paper.

\section{Preliminaries}\label{sec:preliminaries}

We will start by briefly introducing error-correcting codes for readers not
familiar with coding theory, see e.g.\ \cite{berlekamp2015algebraic} for a
more comprehensive  introduction. In the second part of this section, we list
the basic notation and definitions of simple games and their dimension. Here
we refer the interested reader to \cite{TaZw99}.

\subsection{Error-Correcting Codes}
\label{subsec_error_correcting}

The \emph{Hamming weight} $hw(x)$ of a bit vector $x = x_1x_2 \ldots x_n \in\{0,1\}^n$ is the number of
$1$-bits in $x$: $hw(x) = |\{i: x_i = 1\}|$. The \emph{Hamming distance}
$d(x,y)$ between two bit vectors $x$ and $y$ is the number of bit
positions, where the bits in $x$ and $y$ are different: $d(x,y) =
|\{i: x_i \neq y_i\}|$.

Imagine a situation in which a $4$ bit message has to be transmitted
from a sender to a receiver in a noisy environment, where bits are
risking to be flipped during the transmission. By adding extra bits
to the message in a clever way, we can recover the original message
if a few bits are flipped. One way of doing this is by using the well-known
Hamming[$8$,$4$] code, where $4$ bits are added as illustrated
by the following example:

\begin{example} \label{ex:hamming}
The Hamming[$8$,$4$] code is essentially the following set $\mathcal{H}$ of bit
vectors:
$$\mathcal{H} = \{0000\ 0000 , 0001\ 1110, 0010\ 0111, 0011\ 1001,$$
$$0100\ 1011, 0101\ 0101, 0110\ 1100, 0111\ 0010,$$
$$1000\ 1101, 1001\ 0011, 1010\ 1010, 1011\ 0100,$$
$$1100\ 0110, 1101\ 1000, 1110\ 0001, 1111\ 1111\}$$
The set $\mathcal{H}$ contains $16$ vectors -- one vector for each possible $4$
bit message, where the message is the first $4$ bits of a vector. The
$4$ extra bits make it possible to recover a message when bits are
flipped.

The Hamming distance between any two vectors in $\mathcal{H}$ is at least
$4$. This means that we can recover a message if one bit is flipped
by locating the only vector in $\mathcal{H}$ with Hamming distance $1$ to
the received message. If two bits are flipped, we can only detect
that something bad has happened. This is a so-called single-error
correcting and double-error detecting code -- a SECDED code.

Let $\mathcal{C}_8 = \mathcal{H} \setminus \{0000\ 0000 , 1111\
1111\}$ denote the subset of $\mathcal{H}$ consisting of the $14$
bit vectors with Hamming weight $4$. The code $\mathcal{C}_8$ is
referred to as a {\em constant weight code}, since all the members of
$\mathcal{C}_8$ has the same Hamming weight. We will refer to
$\mathcal{C}_8$ several places in the paper.
\end{example}

\subsection{Simple Games and their Dimension}
\label{subsec_simple_games}

A simple game $\Gamma=(N, W)$ is a pair where $N = \{1,\dots,n\}$, for
some positive integer $n$, denotes the set of players or agents and
$W$ is a collection of subsets of $N$, i.e., $W\subseteq 2^N$, satisfying
the following conditions:
\begin{enumerate}
  \item[(1)] $\emptyset\notin W$;
  \item[(2)] $N\in W$;
  \item[(3)] $S\subseteq T\subseteq N$ and $S\in W$ implies $T\in W$.
\end{enumerate}
A {\em coalition} $S$ is a subset of $N$. If $S\in W$, then
it is called winning; otherwise, it is said to be losing.

The relation to a Boolean function $\chi\colon \{0,1\}^n\rightarrow \{0,1\}$ is given as follows:
Let $S$ be the set of coordinates of the input vector $x$ that are equal to $1$, i.e., all players
that vote ``yes''. The players in $N\backslash S$ vote ``no''. If $\chi(x)=1$, then $S$ is winning;
otherwise, it is losing.

Conditions (1) and (2) ensure that the group decision does not contradict the individual preferences
in the case of unanimity. The monotonicity condition (3) models the assumption that an enlarged set of supporters
should not turn the group decision from an acceptance into rejection, which is quite reasonable. So a simple
game $\Gamma$ corresponds to a monotone Boolean function $\chi$ with the extra conditions $\chi(\mathbf{0})=0$ and
$\chi(\mathbf{1})=1$.

Clearly, a simple game $\Gamma$ is uniquely characterized by either its set $W$ of winning or its set $L$ of losing
coalitions, which may both be as large as $2^n-1$ in general. A first reduction is possible: A coalition $S$ is called
\emph{minimal winning} if it is winning and all of its proper subsets are losing. Similarly, a coalition $T$ is called
\emph{maximal losing} if it is losing and all of its proper supersets are winning.
The family consisting of all minimal winning coalitions is denoted by $W^m$ and the family of all maximal losing coalitions
is denoted by $L^M$. Since no minimal winning coalition is a proper subset of another minimal winning coalition, we can
apply Sperner's Lemma, see e.g.~\cite{Lubell1966}, to conclude $\left|W^m\right|\le {n \choose {\lfloor n/2\rfloor}}$.
Similarly, we conclude $\left|L^M\right|\le {n \choose {\lfloor n/2\rfloor}}$.

A simple game
$\Gamma=(N,W)$ is weighted if there exists a \emph{quota} $q \in \mathbb{R}_{>0}$ and
\emph{weights} $w_1, w_2, \ldots , w_n \in \mathbb{R}_{\ge 0}$ such that $S \in W$
if and only if $\sum_{i \in S} w_i \geq q$. We remark that one can require the weights
and the quota to be non-negative integers~\cite{Freixas1997}. The intersection
$(N, W_1) \cap (N, W_2)$ of two simple games is the simple game $(N, W_1 \cap W_2)$.
Taylor and Zwicker~\cite{TaZw99} have shown that any simple game can be written
as the intersection of $|L^M|$ weighted games $\Gamma_T$, $T \in L^M$, where a coalition
$S$ wins in $\Gamma_T$ if $S \cap (N\setminus T) \neq \emptyset$. A weighted representation
using weights $0$ and $1$ is given as follows: A player in $N \setminus T$ has weight
$1$ and all other players have weight $0$ in the game $\Gamma_T$ that has quota $1$.

The dimension $d$ of a simple game $\Gamma$ is the smallest positive
integer such that $\Gamma = \cap_{i = 1}^d \Gamma_d$, where the games
$\Gamma_i$, $i \in \{ 1, 2, \ldots, d\}$, are weighted. From the previous
considerations we conclude
\begin{equation}
\label{eq:d-upperbound}
d \leq |L^M| \leq \min \left(2^n - |W|, {n \choose \lfloor \frac{n}{2} \rfloor}\right) \enspace .
\end{equation}
To give an intuition of how this upper bound relates to $2^n$, we
can use the the following double inequality that holds for all even
positive integers $n$~\cite{Luke69}:
\begin{equation}
\label{eq:sperner}
\sqrt{\frac{2}{\pi n}} \left(1-\frac{1}{4n}\right) 2^n \leq {n \choose \frac{n}{2}} \leq \sqrt{\frac{2}{\pi n}} \left(1-\frac{2}{9n}\right) 2^n .
\end{equation}
For all odd positive integers $n$, we can use the equality ${n \choose \lfloor \frac{n}{2} \rfloor} =
{n-1 \choose \frac{n-1}{2}} \frac{2n}{n+1}$ and obtain the
following inequalities:
\begin{equation}
\label{eq:sperner2}
{n \choose \lfloor \frac{n}{2} \rfloor} \geq \frac{n}{n+1} \sqrt{\frac{2}{\pi (n-1)}} \left(1-\frac{1}{4(n-1)}\right) 2^{n} \enspace ,
\end{equation}
\begin{equation}
\label{eq:sperner3}
{n \choose \lfloor \frac{n}{2} \rfloor} \leq  \frac{n}{n+1}\sqrt{\frac{2}{\pi (n-1)}} \left(1-\frac{2}{9(n-1)}\right) 2^{n} \enspace .
\end{equation}

For a bit vector $x = x_1x_2 \ldots x_n \in \{0, 1\}^n$ with $n$
bits, we let $S_x$ be the coalition where $i \in S$ if and only $x_i
= 1$. For a coalition $S \subseteq N$, we define the bit vector $x_S$
accordingly. We use the notation $\bar{x}$ and $\bar{S}$ for
complements for bit vectors and sets, respectively.

\section{The Example of Taylor and Zwicker}
\label{TZEx}

Let us reconsider the construction of a simple game with large dimension from \cite[Theorem 1.7.5]{TaZw99}.
To this end, let $k$ be an odd integer, \highlight{$S=\{1,\dots,k\}$}, \highlight{$T=\{k+1,\dots,2k\}$}, and $N=S\cup T$. A coalition
$\highlight{X}\subseteq N$ is winning iff either $|X|\ge k+1$ or $|X|=k$ and $|X\cap T|\equiv 0\pmod 2$.
Denote the corresponding simple game by $\Gamma_k$. The minimal winning coalitions of $\Gamma_k$ are given by $W^m=$
$$
  \Big\{X_1\cup X_2\mid X_1\subseteq S, X_2\subseteq T, \left|X_2\right|\equiv 0\!\!\!\!\pmod 2, \left|X_1\cup X_2\right|=k\Big\}
$$
and the maximal losing coalitions of $\Gamma_k$ are given by $L^M=$
$$
  \Big\{X_1\cup X_2\mid X_1\subseteq S, X_2\subseteq T, \left|X_2\right|\equiv 1\!\!\!\!\pmod 2, \left|X_1\cup X_2\right|=k\Big\}.
$$
Since $k\equiv 1\pmod 2$, we have $n\equiv 2\pmod 4$ for $n=2k=|N|$ and $\left|W^m\right|=\left|L^M\right|=
\frac{1}{2}\cdot {n\choose n/2}$, so that the dimension of $\Gamma_k$ is at most $\frac{1}{2}\cdot {n\choose n/2}$. We remark
that $\Gamma_k$ is self-dual, so that its dimension equals its co-dimension.

\begin{theorem}
  \label{thm_ex_TZ}
   For each odd integer $k$, the dimension of $\Gamma_k$ is given by $2^{k-1}$.
\end{theorem}
\begin{proof}
  Let $\mathcal{C}=\left\{x\overline{x}\,:\,x\in\{0,1\}^k,\sum_{i=1}^k x_i\equiv 0\pmod 2\right\}$,
  where $\overline{x}$ denotes the negation of a binary vector and $xy$ denotes the concatenation of
  two binary vectors $x$ and $y$. We have $\mathcal{C}= L^M$, %$\mathcal{C}\subseteq L^M$, 
  $\left|\mathcal{C}\right|=2^{k-1}$,
  and we remark that the minimum Hamming distance of $\mathcal{C}$ is $4$ for $k>1$.

  For the lower bound on the dimension, we refer to \cite[Theorem 1.7.5]{TaZw99}.\footnote{Using the general
  approach and notation of \cite{Kurz2015} we can state a quick proof: For each $x,y\in\mathcal{C}$ with
  $x\neq y$ there exist indices $1\le i\le k$, $k+1\le j\le 2k$ with $x_i\neq y_i$, $x_j\neq y_j$, and $x_i\neq x_j$.
  Negating $x_i$, $x_j$, $y_i$, and $y_j$ gives two winning vectors $x'$, $y'$ with $x+y=x'+y'$, i.e., we have determined
  a $2$-trade, so that the dimension is at least $2^{k-1}$.}

  For the other direction set $\mathcal{C}^p=\left\{x\in\{0,1\}^k\,:\,\sum_{i=1}^k x_i\equiv 0\pmod 2\right\}$.
  Since $2^{1-1}=\frac{1}{2}\cdot {2\choose 1}$, we can assume $k\ge 3$. We set
  $v=\cap_{x\in \mathcal{C}^p} v_x$, where $v_x=\left[q^x;w_1^x,\dots,w_{2k}^x\right]$ with
  \begin{itemize}
    \item $w_i^x=\left\{\begin{array}{rcl} \highlight{0} &:& x_i=\highlight{1},\\\highlight{2}&:&x_i=0\end{array}\right.$
          for all $1\le i\le k$, $w_j^x=\highlight{1}$ for all $k+1\le j\le 2k$, and $q^x=\highlight{k}-(hw(x)-1)$ if $x\neq \mathbf{0}$;
    \item $w_i^x=\highlight{1}$ for all $1\le i\le k$, $w_j^x=\highlight{0}$ for all $k+1\le j\le 2k$, and $q^x=\highlight{1}$ if $x=\mathbf{0}$.
  \end{itemize}
  Let $S_1\subseteq N$ with $|S_1|\ge k+1$. For each $x\in \mathcal{C}^p\highlight{\backslash\{\mathbf{0}\}}$, we have
  \highlight{$w^x(S_1)\ge hw(x)\cdot 0 +(k+1-hw(x))\cdot 1= q^x$}. \highlight{Since $\left|S_1\cap S\right|\ge 1$, we additionally
  have $w^{\mathbf{0}}(S_1)\ge 1=q^{\mathbf{0}}$,} so that $S_1$ is winning in $v$. Now let $S_2$ be a coalition with $|S_2|=k$ and $\left|S_2\cap S\right|\equiv 0\pmod 2$.
  If $x$ is the characteristic vector of $S_2\cap S$, then
  \begin{itemize}
    \item $w^x(S_2)=hw(x)\cdot \highlight{0}+(k-hw(x))\cdot \highlight{1}=\highlight{k}-hw(x)<q^x$ for $x\neq \mathbf{0}$;
    \item $w^{x}(S_2)=\highlight{0<1}=q^x$ for $x=\mathbf{0}$,
  \end{itemize}
  so that $S_2$ is a losing coalition in $v$. Let $S_3$ be a coalition with $|S_3|=k$ and
  $\left|S_3\cap S\right|\equiv 1\pmod 2$. \highlight{Since $\left|S_3\cap S\right|\ge 1$ we have $w^{\mathbf{0}}(S_3)\ge 1=q^{\mathbf{0}}$.}
  \highlight{Now let $x\in\mathcal{C}^p\backslash\{\mathbf{0}\}$ be arbitrary. If $\left|S_3\cap S\right|<hw(x)$, then we have
  $w^x(S_3)\ge (hw(x)-1)\cdot 0+(k-hw(x)+1)\cdot 1=q^x$. If $\left|S_3\cap S\right|>hw(x)$, then }
  there exists a player $i\in S_3\cap S$ with $w_i^x=2$, so that $w^x(S_3)\ge hw(x)\cdot \highlight{0}+\highlight{2}+(k-hw(x)\highlight{-1})
  \cdot \highlight{1}=q^x$. \highlight{Thus,} $S_3$ is winning in $v$.
  Finally, let $S_4$ be a coalition of cardinality $k-1$.
  Since $k-1$ is even, we have the following two cases:
  \begin{itemize}
    \item $\left|S_4\cap S\right|\equiv 0\pmod 2$, $\left|S_4\cap T\right|\equiv 0\pmod 2$,
    \item $\left|S_4\cap S\right|\equiv 1\pmod 2$, $\left|S_4\cap T\right|\equiv 1\pmod 2$.
  \end{itemize}
  In both cases, it is possible to extend $S_4$ to a coalition $S_5\in\mathcal{L}^M$ by adding a player, so that
  $S_4$ has to be losing in $v$. Thus, we have $v=\Gamma_k$ and $\dim(\Gamma_k)\le 2^{k-1}$.
\end{proof}

Now let us restate the example of \cite[Theorem 8]{ElkindGGW08}:
Let $k$ be an odd integer and $n=2k$, $N=\{1,\dots,n\}$. Consider the simple game where all coalitions of cardinality
larger than $k$ are winning and all coalitions of cardinality smaller than $k$ are losing. A coalition $X$ of cardinality
$k$ is winning iff the Hamming distance between $X$ and $\{1,\dots,k\}$ is equivalent to $2$ modulo $4$. In other words,
this means that $\left|X\cap \{1,\dots,k\}\right|$ is even and $\left|X\cap \{k+1,\dots,n\}\right|$ is odd.

Interchanging the first $k$ players with the last $k$ players yields the example of Taylor and Zwicker.
Since Theorem~8 in \cite{ElkindGGW08} claims that the dimension is at least ${{2k}\choose k}/2$,
there is a contradiction to Theorem~\ref{thm_ex_TZ}. The flaw\footnote{We would like to thank Edith
Elkind for directly pointing to the position where the proof
breaks down in a private communication.} of
the corresponding proof happens where it says that if $x$ is
the bit vector of a losing coalition and $x_i\neq x_j$, then
switching $x_i$ and $x_j$ results in a bit vector of a
winning coalition. An explicit counter example for $n=6$ is given by
the characteristic vectors $100110$ and $010110$ which both
represent losing coalitions.

\section{From Error Correcting Codes to Simple Games}\label{sec:construction}

In this section, we present a generic recipe for constructing the
simple games forming the basis for our results. Throughout the
paper, we let $\mathcal{C} \subseteq \{0, 1\}^n$ denote a set of bit
vectors of length $n$ having positive Hamming weight satisfying this
condition:
\begin{equation}
\label{eq:C}
\forall x \neq y \in \mathcal{C}: |hw(x)-hw(y)| < d(x,y)-2
%\forall x, y \in \mathcal{C}: x \neq y \Rightarrow |hw(x)-hw(y)| < d(x,y)-2
\end{equation}

For $x \in \mathcal{C}$, we define the simple game $\Gamma_x$ with
players $N = \{1, 2, \ldots, n\}$ as follows: $S$ wins in $\Gamma_x$ if and only if $S
\cap S_x \neq \emptyset$. The simple game $\Gamma_{\mathcal{C}}$ is
now defined by $\Gamma_{\mathcal{C}} = \cap_{x \in \mathcal{C}}
\Gamma_x$. In other words, a set $S$ is winning if and only if $S$
is a so-called {\em hitting set} for the collection of sets
$\{S_x\}_{x \in \mathcal{C}}$.

The error-correcting code $\mathcal{C}_8$ from
Example~\ref{ex:hamming} is a set of bit vectors satisfying
(\ref{eq:C}). Another example is the following:
\begin{example} \label{ex:game}
Let $\mathcal{C}$ be defined as follows for $n=8$: $$ \mathcal{C} = \{0000\
1111, 1100\ 0000, 0011\ 1100\}$$ The Hamming weights of the vectors
$0000\ 1111$ and $1100\ 0000$ differ by $2$ but their Hamming
distance is $6$. So $(\ref{eq:C})$ holds for these vectors.
Coalition $\{1, 5\}$ is winning in
$\Gamma_{\mathcal{C}}$, since it intersects the sets $\{5,6,7,8\}$,
$\{1, 2\}$ and $\{3, 4, 5, 6\}$. The bit vector $1000\ 1000$ that
corresponds to the set $\{1, 5\}$ shares at least one $1$-bit with
all members of $\mathcal{C}$. \end{example}

\subsection{A Dimension Lemma}

We now prove a lemma explicitly stating the dimension of our games.

\begin{lemma}\label{lem:dim} The dimension of $\Gamma_{\mathcal{C}}$ is
$|\mathcal{C}|$. \end{lemma}

\begin{proof} The game $\Gamma_x$, $x \in \mathcal{C}$, is clearly
weighted, so the dimension of $\Gamma_{\mathcal{C}}$ is not higher
than $|\mathcal{C}|$.

We now assume that the dimension of $\Gamma_{\mathcal{C}}$ is less
than $|\mathcal{C}|$. Let $L_x = N \setminus S_{x}$ for $x \in
\mathcal{C}$. The coalition $L_x$ is clearly a losing coalition in
$\Gamma_{\mathcal{C}}$ because $L_x \cap S_x = \emptyset$. Using the
pigeonhole principle, we conclude that there are $x,y \in
\mathcal{C}$ with $x \neq y$ such that $L_x$ and $L_y$ lose in the
same weighted game $\Gamma^\prime$, where $\Gamma^\prime$ is one of
the less than $|\mathcal{C}|$ weighted games whose intersection is
$\Gamma_{\mathcal{C}}$.

By considering basic properties for the Hamming distance and the
Hamming weight, we observe that (\ref{eq:C}) also holds if we replace
$x$ and $y$ with their complements $\bar{x}$ and $\bar{y}$. If one
of the vectors $\bar{x}$ or $\bar{y}$ had all $1$-bits in the
$d(\bar{x},\bar{y})$ positions, where the two vectors differ, then the
left-hand side of (\ref{eq:C}) would be $d(\bar{x},\bar{y})$ and
(\ref{eq:C}) would not hold. We therefore conclude that there are
players $p_x \in L_x \setminus L_y$ and $p_y \in L_y \setminus L_x$.
We let $A$ and $B$ be the coalitions obtained if $L_x$ and $L_y$
swap these players: $A= (L_x \setminus \{ p_x \}) \cup \{ p_y \}$
and $B = (L_y \setminus \{ p_y \}) \cup \{ p_x \}$.

We now show that $A$ and $B$ are winning coalitions in
$\Gamma_{\mathcal{C}}$. Without loss of generality, we consider the
coalition $A$. It is clear that $x_A$ and $x$ share a $1$-bit so $A$
wins in $\Gamma_x$. Now let us assume that there is a member $z$ of
$\mathcal{C} \setminus \{x\}$ such that $A$ loses in $\Gamma_z$. In
other words, $x_A$ and $z$ do not share a $1$-bit. The vector $x_A$
is obtained by flipping a $0$-bit and a $1$-bit in the vector
$\bar{x}$: \begin{equation} \label{eq:L1_1} d(x_{A}, \bar{z}) \geq
d(\bar{x}, \bar{z}) - 2 \enspace . \end{equation} The $d(x_A,
\bar{z})$ bits shared by $x_A$ and $z$ are all $0$ in which case we
have the following: \begin{equation} \label{eq:L1_2} d(x_A, \bar{z})
+ hw(x_A) + hw(z) = n \enspace . \end{equation} We now use $hw(x_A)
= n - hw(x)$ together with (\ref{eq:L1_2}): \begin{equation}
\label{eq:L1_3} d(x_A, \bar{z}) = hw(x)-hw(z) \enspace .
\end{equation} By using $d(x, z) = d(\bar{x}, \bar{z})$ and
(\ref{eq:L1_1}) and (\ref{eq:L1_3}), we obtain the following
inequality: \begin{equation} \label{eq:L1_4} hw(x)-hw(z) \geq d(x,
z) - 2 \enspace . \end{equation} Since (\ref{eq:L1_4}) contradicts
(\ref{eq:C}), we conclude that $A$ wins in $\Gamma_z$ for any $z \in
\mathcal{C}$. Consequently, $A$ also wins in $\Gamma_{\mathcal{C}}$.

Summing up, we now have two coalitions $L_x$ and $L_y$ that lose in
$\Gamma^\prime$, and we can obtain two winning coalitions in
$\Gamma_{\mathcal{C}}$ if $L_x$ and $L_y$ swap two players. These
coalitions also win in $\Gamma^\prime$ and we obtain a contradiction,
since this would mean that the total weight in $\Gamma^\prime$ of
the players in $L_x$ and $L_y$ has increased.
\end{proof}
It is worth noting that the dimension of the game
$\Gamma_{\mathcal{C}}$ is $|L^M|$ since $L^M = \{ L_x \}_{x \in \mathcal{C}}$.

If we can construct games with
dimension $m$ using our approach, we can also construct games with
dimension $m'$ for every $m' \leq m$ as expressed by the following
corollary:

\begin{corollary}
Let $\Gamma_\mathcal{C}$ be a simple game
with $n$ players and dimension $m$, then there are simple games
with $n$ players and dimension $m'$, $1 \leq m' \leq m$.
\end{corollary}

\begin{proof}
Just delete some elements from $\mathcal{C}$.
\end{proof}

%Combinatorial techniques usually invoke the so-called \emph{trades}.
%A \emph{trading transform} for a simple game is a collection of coalitions
%$J = <S_1, \ldots, S_j, T1, \ldots, T_j>$ such that $|{h:i\in{}S_h}| = |{h:i\in{}T_h}|$,
%$i\in{}N$. An \emph{$m$-trade} is a trading transform with $j\le{}m$ such that all $S_h$
%are winning and all $T_h$ are losing coalitions.
%The existence of a 2-trade $<S_1,S_2;T_1,T_2>$ implies that the game cannot be weighted.

%In the same vein, the class of \emph{complete simple games} is exactly the subclass of simple
%games that are $2$-trade (i.e., swap-robust)~\cite{TaZw99}.
%\begin{corollary}
%The games $\Gamma_\mathcal{C}$ are complete if and only if $|\mathcal{C}|=1$.
%\end{corollary}

%\begin{proof}
%If $|\mathcal{C}|=1$ then the game is weighted and clearly complete.
%If $|\mathcal{C}|>1$ then there exists a "swap" (see proof of Lemma~\ref{lem:dim}.
%\end{proof}

\section{Simple Games with High Dimension}\label{sec:hdproof}

The key question we will deal with in this section is the following:
Can we find families $\mathcal{C}$ of bit vectors with high
cardinality satisfying (\ref{eq:C})? According to
Lemma~\ref{lem:dim}, this would automatically give us games with
high dimension. From the theory on error-correcting codes, we know
how to construct relatively large families of bit vectors forming
SECDED constant weight codes. If we pick such a code, we clearly have
a family $\mathcal{C}$ satisfying (\ref{eq:C}). This observation is
the basis for the proofs in this section. As an example, the code
$\mathcal{C}_8$ from Example~\ref{ex:hamming} corresponds to a
simple game with $8$ players and dimension $14$.

It is important to stress that constant weight SECDED codes are not
the only families satisfying the generic recipe (\ref{eq:C}) as
illustrated by Example~\ref{ex:game}. There are many other families
that satisfy (\ref{eq:C}), but we will use constant weight SECDED
codes to construct our games with high dimension. In other words,
there might be families with larger cardinalities compared to
constant weight SECDED codes satisfying (\ref{eq:C}).

Agrell et al. \cite{Agrell00} present lower bounds for
cardinalities of constant weight SECDED codes. These lower bounds
can be directly translated to lower bounds for dimensions for simple
games if we use Lemma~\ref{lem:dim}. This allows us to set up
Table~\ref{tab:bounds} that compares the dimensions of the games
produced using composition of unanimity games~\cite{Freixas2001265}
with the dimensions of the games based on our approach and the lower
bounds from \cite{Agrell00}. The first column displays $n$. The
second column presents the dimensions of the games from
\cite{Freixas2001265} and~\cite{TaZw99}. The third column contains
the dimensions of the games produced u\-sing our approach and constant
weight SECDED codes. Finally, the last column shows the, slightly improved, upper bound
${{n}\choose{{\lfloor n/2\rfloor}}}-1$.\footnote{Sperner's Theorem also classifies the cases
where his bound is tight. Since all of the corresponding simple games are indeed
weighted, the previous upper bound can be reduced by $1$.} As an example, we can see that
our approach leads to a simple game with dimension $14$ for $n=8$ --
the game $\Gamma_{\mathcal{C}_8}$.
\begin{table}
\centering
\caption{A comparison of the dimensions of the games
produced using composition of unanimity games and the dimensions of the games based on our approach.}
%\begin{center}
\begin{tabular}{|r|c|c|c|}
\hline
$n$ & Unanimity games & Our approach & ${{n}\choose{{\lfloor n/2\rfloor}}}-1$ \\
\hline
6 &  4 & 4 &  19\\
7 &  4 & 7 &  34\\
8 &  8 & 14 &  69\\
9 &  9 &  18 &  125\\
10 &  16 &  36 &  251\\
11 &  18 &  66 &  461 \\
12 &  32 &  132 &  923\\
13 &  36 &  166 &  1715\\
14 &  64 &  325 &  3431\\
15 &  81 &  585 &  6434\\
16 &  128 &  1170 &  12869\\
17 &  162 &  1770 &  24309\\
18 &  256 &  3540 &  48619\\
19 &  324 &  6726 &  92377\\
20 &  512 &  13452 &  184755\\
\hline
\end{tabular}
%\end{center}
\label{tab:bounds}
\end{table}

We are now ready to consider all other values of $n$. Initially, we
consider the case where $n$ is a power of $2$. The following lemma
generalizes the example described earlier
with $|\mathcal{C}_8|= 14$ for $n=8$ to $n=2^m$ for $m \geq 3$.
\begin{lemma}
\label{lem:pow_2}
Let $n = 2^m$ where $m$ is an integer, $m \geq 3$. There is a set of
bit vectors $\mathcal{C} \subseteq \{0, 1\}^n$ satisfying
(\ref{eq:C}) with
\begin{equation}
\label{eq:lem_pow_2}
|\mathcal{C}| = \frac{2}{n} \left( \frac{1}{2}{n \choose \frac{n}{2}} + (n-1){\frac{n}{2}-1 \choose \frac{n}{4}} \right) \enspace .
\end{equation}
\end{lemma}

\begin{proof}
Let $t = 2^m - 1$. The enumerator polynomial for an error-correcting
code is a polynomial, where the $i$'th coefficient, $a_i$, is the
number of bit vectors of Hamming weight $i$. According to
\cite{Wicker1994}, the enumerator polynomial for the well-known
Hamming[$t$,$t-m$] code that contains bit vectors of length $t$ is:
$$A(x) = \frac{(1+x)^t+t(1-x)(1-x^2)^{(t-1)/2}}{t+1} \enspace .$$

Let $i = \frac{t-1}{2} $ ($i= 2^{m-1}-1$ is odd and $i+1$ is even):
$$a_i = \frac{1}{t+1} \left( {t \choose i} + t(-1)^{\frac{i+1}{2}}{i \choose \frac{i-1}{2}} \right)$$
$$a_{i+1} = \frac{1}{t+1} \left( {t \choose i+1} + t(-1)^{\frac{i+1}{2}}{i \choose \frac{i+1}{2}} \right) = a_i$$
The extended code Hamming[$t+1$,$t-m$] is a SECDED code. We can now
let $\mathcal{C}$ be the subset of the extended code containing the
bit vectors with Hamming weight $\frac{n}{2}$. This is a constant
weight SECDED code satisfying (\ref{eq:C}).

Set $n = t+1 = 2i + 2$. The number of bit vectors in the extended
code with Hamming weight $\frac{n}{2} = i+1$ is $a_i + a_{i+1} =
2a_{i+1}$:
$$ 2a_{i+1} = \frac{2}{n} \left( {n-1 \choose \frac{n}{2}} + (n-1)(-1)^{\frac{n}{4}}{\frac{n}{2}-1 \choose \frac{n}{4}} \right) $$

For $n \geq 8$, we have:
$$ 2a_{i+1} = \frac{2}{n} \left( {n-1 \choose \frac{n}{2}} + (n-1){\frac{n}{2}-1 \choose \frac{n}{4}} \right) \enspace .$$

We now use:
$$  {n \choose \frac{n}{2}} = {n-1 \choose \frac{n}{2}} + {n-1 \choose \frac{n}{2}-1} = 2 {n-1 \choose \frac{n}{2}} $$
to obtain %%and obtain
$$ 2a_{i+1} = \frac{2}{n} \left( \frac{1}{2}{n \choose \frac{n}{2}} + (n-1){\frac{n}{2}-1 \choose \frac{n}{4}} \right) \enspace .$$
\end{proof}

We now state our main theorem, where we also consider values of $n$ that are not powers of $2$.

\begin{theorem}
 \label{thm:main}
 For any positive integer $n$ there is a simple game with $n$ players and dimension $d$ satisfying:
\begin{equation}
\label{eq:thm_generic}
d \geq \frac{1}{n} {n \choose \lfloor \frac{n}{2} \rfloor} \enspace .
\end{equation}
If $n = 2^m$ for an integer $m \geq 3$, then there is a simple game with $n$ players and dimension $d$ such that
\begin{equation}
\label{eq:thm_pow_2}
d = \frac{1}{n}{n \choose \frac{n}{2}} + \frac{2(n-1)}{n}{\frac{n}{2}-1 \choose \frac{n}{4}} \enspace .
\end{equation}
\end{theorem}

\begin{proof}
Graham and Sloane~\cite{Graham80} have shown that there is
constant weight SECDED code with Hamming weight $w$ with cardinality
at least $\frac{1}{n}{n \choose w}$ for any $w$. For $w = \lfloor
\frac{n}{2} \rfloor$, we get (\ref{eq:thm_generic}) by using
Lemma~\ref{lem:dim}. Lemma~\ref{lem:dim} and Lemma~\ref{lem:pow_2}
give us (\ref{eq:thm_pow_2}).
\end{proof}
It follows from (\ref{eq:sperner}) and (\ref{eq:sperner2}) that the
lower bound presented in Theorem~\ref{thm:main} is $2^{n-o(n)}$. Our
games are easily seen to be within a factor $n$ from the upper bound
from (\ref{eq:d-upperbound}). Finally, we point out that the proof of the
lower bound in \cite{Graham80} is constructive.

\section{Conclusion}
We have presented a link from coding theory to the dimension of simple games. We are
not aware of any other connection between coding theory and simple games. While it seems
a rather tough problem to determine the exact dimension of a simple game, 
we have provided an exact formula for those simple games arising from error correcting codes in Lemma~\ref{lem:dim}.
Via this connection, any improvement on lower bounds of constant weight codes improves the
stated lower bounds for the worst-case dimensions of simple games. For the other direction, it
would be interesting to know whether unrestricted codes satisfying inequality~(\ref{eq:C}) have
some application in coding theory. Till now, it is even unclear, at least for us, if those codes can
be strictly larger than constant weight codes. From our point of view, this connection should be explored
in more detail.

The asymptotic magnitude of the worst-case examples with respect to the dimension of simple games is determined,
which closes a gap in the literature and uncovers a flaw in a previous attempt. The bad news is that from a
representation complexity point of view, the usage of intersections of weighted games cannot be a solution
for all cases. From a practical point of view, one may nevertheless ask whether the set of weighted games with
\textit{small} dimension are not \textit{too far apart} from the set of simple games, so that there is
no reason to use high-dimensional simple games in reality.

From a mathematical point of view, it would be interesting to determine the exact values of the worst-case examples.

The construction of the representing weighted games is still widely open and deserves further attention.

\end{document}